\documentclass[11pt,letterpaper]{article}

\usepackage{xspace}
\usepackage[numbers]{natbib}
\usepackage{multirow}
\usepackage{amsmath,amssymb,amsthm}
\usepackage[margin=1in]{geometry}

\theoremstyle{plain}
\newtheorem{proposition}{Proposition}
\newtheorem{theorem}{Theorem}
\newtheorem{lemma}{Lemma}

\theoremstyle{definition}
\newtheorem{definition}{Definition}
\newcommand{\E}{{}\mathop{\rm E}}

\newcommand{\thmref}[1]{Theorem~\ref{#1}}
\newcommand{\lemref}[1]{Lemma~\ref{#1}}
\newcommand{\propref}[1]{Proposition~\ref{#1}}

\newcommand{\defref}[1]{Definition~\ref{#1}}
\newcommand{\appref}[1]{Appendix~\ref{#1}}

\newcommand{\reals}{\mathbb{R}}
\newcommand{\SW}{\text{SW}\xspace}
\newcommand{\OPT}{\text{OPT}\xspace}
\newcommand{\PoA}{\text{PoA}\xspace}
\newcommand{\OS}{\text{OS}\xspace}
\newcommand{\OXS}{\text{OXS}\xspace}
\newcommand{\GS}{\text{GS}\xspace}
\newcommand{\SM}{\text{SM}\xspace}
\newcommand{\XOS}{\text{XOS}\xspace}
\newcommand{\CF}{\text{CF}\xspace}
\newcommand{\DSE}{\text{DSE}\xspace}
\newcommand{\PNE}{\text{PNE}\xspace}
\newcommand{\MNE}{\text{MNE}\xspace}
\newcommand{\CE}{\text{CE}\xspace}
\newcommand{\CCE}{\text{CCE}\xspace}

\newcommand{\pgraph}[1]{\paragraph{#1.}}

\begin{document}

\title{Valuation Compressions in VCG-Based Combinatorial Auctions}

\author{
  Paul D\"utting\thanks{%
	Department of Computer Science, 
	Stanford University, 
  	353 Serra Mall, 
	Stanford, CA 94305-9045, USA. 
	Email: \texttt{paul.duetting@stanford.edu}. 
	This work was supported by an SNF Postdoctoral Fellowship.} 
  \and 
  Monika Henzinger\thanks{%
	Faculty of Computer Science, 
	University of Vienna, 
	W\"ahringer Stra{\ss}e 29, 
	1090 Wien, Austria. 
	Email: \texttt{\{monika.henzinger,martin.starnberger\}@univie.ac.at}.
	This work was funded by the Vienna Science and Technology Fund (WWTF) through project ICT10-002, and by the University of Vienna through IK I049-N.}
  \and 
  Martin Starnberger\footnotemark[2]
}

\date{}

\maketitle

\begin{abstract}
The focus of classic mechanism design has been on truthful direct-revelation mechanisms. In the context of combinatorial auctions the truthful direct-revelation mechanism that maximizes social welfare is the VCG mechanism. For many valuation spaces computing the allocation and payments of the VCG mechanism, however, is a computationally hard problem. We thus study the performance of the VCG mechanism when bidders are forced to choose bids from a subspace of the valuation space for which the VCG outcome can be computed efficiently. We prove improved upper bounds on the welfare loss for restrictions to additive bids and upper and lower bounds for restrictions to non-additive bids. These bounds show that the welfare loss increases in expressiveness. All our bounds apply to equilibrium concepts that can be computed in polynomial time as well as to learning outcomes. 
\end{abstract}

\section{Introduction}
An important field at the intersection of economics and computer science is the field of mechanism design. The goal of mechanism design is to devise mechanisms consisting of an outcome rule and a payment rule that implement desirable outcomes in strategic equilibrium. A fundamental result in mechanism design theory, the so-called {\em revelation principle}, asserts that any equilibrium outcome of any mechanism can be obtained as a truthful equilibrium of a direct-revelation mechanism. However, the revelation principle says nothing about the computational complexity of such a truthful direct-revelation mechanism.

In the context of combinatorial auctions the truthful direct-revelation mechanism that maximizes welfare is the \emph{Vickrey-Clarke-Groves (VCG) mechanism} \cite{Vickrey61,Clarke71,Groves73}. Unfortunately, for many valuation spaces computing the VCG allocation and payments is a computationally hard problem. This is, for example, the case for subadditive, fractionally subadditive, and submodular valuations \cite{LehmannLehmannNisan05}. We thus study the performance of the VCG mechanism in settings in which the bidders are forced to use bids from a subspace of the valuation space for which the allocation and payments can be computed efficiently. This is obviously the case for additive bids, where the VCG-based mechanism can be interpreted as a separate second-price auction for each item. But it is also the case for the syntactically defined bidding space \OXS, which stands for ORs of XORs of singletons, and the semantically defined bidding space \GS, which stands for gross substitutes. For \OXS bids polynomial-time algorithms for finding a maximum weight matching in a bipartite graph such as the algorithms of \cite{Tarjan83} and \cite{FredmanTarjan87} can be used. For \GS bids there is a fully polynomial-time approximation scheme due to \cite{KelsoCrawford82} and polynomial-time algorithms based on linear programming \cite{BikchandaniEtAl02} and convolutions of $M^{\#}$-concave functions \cite{Murota00,Murota96,MurotaTamura01}.

One consequence of restrictions of this kind, that we refer to as \emph{valuation compressions}, is that there is typically no longer a truthful dominant-strategy equilibrium that maximizes welfare. We therefore analyze the \emph{Price of Anarchy}, i.e., the ratio between the optimal welfare and the worst possible welfare at equilibrium. We focus on equilibrium concepts such as correlated equilibria and coarse correlated equilibria, which can be computed in polynomial time \cite{PapadimitriouRoughgarden08,JiangBrown11}, and naturally emerge from learning processes in which the bidders minimize external or internal regret \cite{FosterVohra97,HartMasColell00,LittlestoneWarmuth94,CesaBianchiEtAl07}.


\pgraph{Our Contribution} 
We start our analysis by showing that for restrictions from subadditive valuations to additive bids deciding whether a pure Nash equilibrium exists is $\mathcal{NP}$-hard. This shows the necessity to study other bidding functions or other equilibrium concepts.

We then define a smoothness notion for mechanisms that we refer to as {\em relaxed smoothness}. This smoothness notion is weaker in some aspects and stronger in another aspect than the weak smoothness notion of \cite{SyrgkanisTardos13}. It is weaker in that it allows an agent's deviating bid to depend on the distribution of the bids of the other agents. It is stronger in that it disallows the agent's deviating bid to depend on his own bid. The former gives us more power to choose the deviating bid, and thus has the potential to lead to better bounds. The latter is needed to ensure that the bounds on the welfare loss extend to coarse correlated equilibria and minimization of external regret.

We use relaxed smoothness to prove an upper bound of $4$ on the Price of Anarchy with respect to correlated and coarse correlated equilibria. Similarly, we show that the average welfare obtained by minimization of internal and external regret converges to $1/4$-th of the optimal welfare. The proofs of these bounds are based on an argument similar to the one in \cite{FeldmanFuGravinLucier13}. Our bounds improve the previously known bounds for these solution concepts by a logarithmic factor. We also use relaxed smoothness to prove bounds for restrictions to non-additive bids. For subadditive valuations the bounds are $O(\log(m))$ resp.~$\Omega(1/\log(m))$, where $m$ denotes the number of items. For fractionally subadditive valuations the bounds are $2$ resp.~$1/2$. The proofs require novel techniques as non-additive bids lead to non-additive prices for which most of the techniques developed in prior work fail. The bounds extend the corresponding bounds of \cite{ChristodoulouKovacsSchapira08,BhawalkarRoughgarden11} from additive to non-additive bids.

Finally, we prove lower bounds on the Price of Anarchy. By showing that VCG-based mechanisms satisfy the {\em outcome closure property} of \cite{Milgrom10} we show that the Price of Anarchy with respect to pure Nash equilibria weakly increases with expressiveness. We thus extend the lower bound of $2$ from \cite{ChristodoulouKovacsSchapira08} from additive to non-additive bids. This shows that our upper bounds for fractionally subadditive valuations are tight. We prove a lower bound of $2.4$ on the Price of Anarchy with respect to pure Nash equilibria that applies to restrictions from subadditive valuations to \OXS bids. Together with the upper bound of $2$ of \cite{BhawalkarRoughgarden11} for restrictions from subadditive valuations to additive bids this shows that the welfare loss can strictly increase with expressiveness.

Our analysis leaves a number of interesting open questions, both regarding the computation of equilibria and regarding improved upper and lower bounds. Interesting questions regarding the computation of equilibria include whether or not mixed Nash equilibria can be computed efficiently for restrictions from subadditive to additive bids or whether pure Nash equilibria can be computed efficiently for restrictions from fractionally subadditive valuations to additive bids. A particularly interesting open problem regarding improved bounds is whether the welfare loss for computable equilibrium concepts and learning outcomes can be shown to be strictly larger for restrictions to non-additive, say \OXS, bids than for restrictions to additive bids. This would show that additive bids are not only sufficient for the best possible bound but also necessary.

\begin{table}[t!]
\centering
\caption{\small Summary of our results (bold) and the related work (regular) for coarse correlated equilibria and minimization of external regret through repeated play. The range indicates upper and lower bounds on the Price of Anarchy.}
\begin{tabular}{cccc}
  \hline
  & & \multicolumn{2}{c}{valuations}\\
  & & less general & subadditive\\
  \hline
  \multirow{2}{*}{bids} 
  & additive     & [2,2]      & [2,\textbf{4}] \\
  & more general & [\textbf{2}, \textbf{2}] & [\textbf{2.4},\textbf{O(log(m))}]\\
\hline
\end{tabular}
\end{table}


\pgraph{Related Work}
The Price of Anarchy of restrictions to additive bids is analyzed in \cite{ChristodoulouKovacsSchapira08,BhawalkarRoughgarden11,FeldmanFuGravinLucier13} for second-price auctions and in \cite{HassidimKaplanMansourNisan11,FeldmanFuGravinLucier13} for first price auctions. The case where all items are identical, but additional items contribute less to the valuation and agents are forced to place additive bids is analyzed in \cite{MarkakisTelelis12,KeijzerMarkakisSchaeferTelelis13}.
Smooth games are defined and analyzed in \cite{Roughgarden09,Roughgarden12}. The smoothness concept is extended to mechanisms in \cite{SyrgkanisTardos13}. 


\section{Preliminaries}\label{sec:prelim}
\pgraph{Combinatorial Auctions}
In a {\em combinatorial auction} there is a set $N$ of $n$ {\em agents} and a set $M$ of $m$ {\em items}. Each agent $i \in N$ employs preferences over bundles of items, represented by a {\em valuation function} $v_i: 2^M \rightarrow \reals_{\ge 0}$. We use $V_i$ for the {\em class of valuation functions} of agent $i$, and $V = \prod_{i \in N} V_i$ for the class of joint valuations. We write $v = (v_i, v_{-i}) \in V$, where $v_i$ denotes agent $i$'s valuation and $v_{-i}$ denotes the valuations of all agents other than $i$. We assume that the valuation functions are {\em normalized} and {\em monotone}, i.e., $v_i(\emptyset) = 0$ and $v_i(S) \leq v_i(T)$ for all $S \subseteq T$.  

A mechanism $M=(f,p)$ is defined by an {\em allocation rule} $f: B \rightarrow \mathcal{P}(M)$ and a {\em payment rule} $p: B \rightarrow \reals^n_{\ge 0}$, where $B$ is the {\em class of  bidding  functions} and $\mathcal{P}(M)$ denotes the {\em set of  allocations} consisting of all possible partitions $X$ of the set of items $M$ into $n$ sets $X_1, \dots, X_n$. As with valuations we write $b_i$ for agent $i$'s bid, and $b_{-i}$ for the bids by the agents other than $i$. We define the {\em social welfare} of an {\em allocation} $X$ as the sum $\SW(X) = \sum_{i\in N} v_i(X_i)$ of the agents' valuations and use $\OPT(v)$ to denote the maximal achievable social welfare. We say that an allocation rule $f$ is {\em efficient} if for all bids $b$ it chooses the allocation $f(b)$ that maximizes the sum of the agent's bids, i.e., $\sum_{i \in N} b_i(f_i(b)) = \max_{X \in \mathcal{P}(M)} \sum_{i \in N} b_i(X_i).$ We assume {\em quasi-linear preferences}, i.e., agent $i$'s {\em utility} under mechanism $M$ given valuations $v$ and bids $b$ is $u_i(b,v_i) = v_i(f_i(b)) - p_i(b).$

We focus on the {\em Vickrey-Clarke-Groves (VCG)} mechanism \cite{Vickrey61,Clarke71,Groves73}. Define $b_{-i}(S) = \max_{X \in \mathcal{P}(S)} \sum_{j\neq i} b_j(X_j)$ for all $S \subseteq M$. The VCG mechanisms starts from an efficient allocation rule $f$ and computes the payment of each agent $i$ as $p_i(b) = b_{-i}(M) - b_{-i}(M\setminus f_i(b))$. As the payment $p_i(b)$ only depends on the bundle $f_i(b)$ allocated to agent $i$ and the bids $b_{-i}$ of the agents other than $i$, we also use $p_i(f_i(b),b_{-i})$ to denote agent $i$'s payment.

If the bids are additive then the VCG prices are additive, i.e., for every agent $i$ and every bundle $S \subseteq M$ we have  $p_i(S,b_{-i}) = \sum_{j \in S} \max_{k \neq i} b_k(j)$. Furthermore, the set of items that an agent wins in the VCG mechanism are the items for which he has the highest bid, i.e., agent $i$ wins item $j$ against bids $b_{-i}$ if $b_i(j) \ge \max_{k \neq i} b_k(j) = p_i(j)$ (ignoring ties). Many of the complications in this paper come from the fact that these two observations do {\em not} apply to non-additive bids.


\pgraph{Valuation Compressions}
Our main object of study in this paper are {\em valuation compressions}, i.e., restrictions of the class of bidding functions $B$ to a strict subclass of the class of valuation functions $V$.\footnote{This definition is consistent with the notion of {\em simplification} in \cite{Milgrom10,DuettingFischerParkes11}.} Specifically, we consider valuations and bids from the following hierarchy due to \cite{LehmannLehmannNisan05},
\[
	\OS \subset \OXS \subset \GS \subset \SM \subset \XOS \subset \CF\enspace,
\]
where OS stands for additive, GS for gross substitutes, SM for submodular, and CF for subadditive.

The classes OXS and XOS are syntactically defined. Define OR ($\vee$) as\linebreak $(u \vee w)(S) = \max_{T \subseteq S}(u(T) + w(S \setminus T))$ and XOR ($\otimes$) as $(u \otimes w) (S) =\linebreak \max(u(S), w(S))$. Define XS as the class of valuations that assign the same value to all bundles that contain a specific item and zero otherwise. Then OXS is the class of valuations that can be described as ORs of XORs of XS valuations and XOS is the class of valuations that can be described by XORs of ORs of XS valuations. 

Another important class is the class $\beta$\text{-XOS}, where $\beta \ge 1$, of $\beta$-fractionally subadditive valuations. A valuation $v_i$ is $\beta$-fractionally subadditive if for every subset of items $T$ there exists an additive valuation $a_i$ such that (a) $\sum_{j \in T} a_i(j) \ge v_i(T)/\beta$ and (b) $\sum_{j \in S} a_i(j) \le v_i(S)$ for all $S \subseteq T$.  It can be shown that the special case $\beta = 1$ corresponds to the class XOS, and that the class CF is contained in $O(\log(m))$\text{-XOS} (see, e.g., Theorem~5.2 in \cite{BhawalkarRoughgarden11}).  Functions in XOS are called {\em fractionally subadditive}.


\pgraph{Solution Concepts}
We use game-theoretic reasoning to analyze how agents interact with the mechanism, a desirable criterion being stability according to some solution concept. In the {\em complete information} model the agents are assumed to know each others' valuations, and in the {\em incomplete information} model the agents' only know from which distribution the valuations of the other agents are drawn. In the remainder we focus on complete information. The definitions and our results for incomplete information are given in Appendix~\ref{app:incomplete}.

The static solution concepts that we consider in the complete information setting are:
\[
	\DSE \subset \PNE \subset \MNE \subset \CE \subset \CCE\enspace,
\] 
where DSE stands for dominant strategy equilibrium, PNE for pure Nash equilibrium, MNE for mixed Nash equilibrium, CE for correlated equilibrium, and CCE for coarse correlated equilibrium. 

In our analysis we only need the definitions of pure Nash and coarse correlated equilibria. Bids $b \in B$ constitute a {\em pure Nash equilibrium (PNE)} for valuations $v \in V$ if for every agent $i \in N$ and every bid $b'_i \in B_i$, $u_i(b_i,b_{-i},v_i) \ge u_i(b'_i,b_{-i},v_i)$. A distribution $\mathcal{B}$ over bids $b \in B$ is a {\em coarse correlated equilibrium (CCE)} for valuations $v \in V$ if for every agent $i \in N$ and every pure deviation $b'_i \in B_i$, $\E_{b \sim \mathcal{B}}[u_i(b_i,b_{-i},v_i)] \ge \E_{b \sim \mathcal{B}}[u_i(b'_i,b_{-i},v_i)]$. 

The dynamic solution concept that we consider in this setting is regret minimization. A sequence of bids $b^1,\dots,b^T$ incurs {\em vanishing average external regret} if for all agents $i$, $\sum_{t=1}^{T} u_i(b^t_i,b^t_{-i},v_i) \ge \max_{b'_i} \sum_{t=1}^{T} u_i(b'_i,b^t_{-i},v_i) - o(T)$ holds, where $o(\cdot)$ denotes the little-oh notation. The empirical distribution of bids in a sequence of bids that incurs vanishing external regret converges to a coarse correlated equilibrium (see, e.g., Chapter 4 of \cite{agtbook}).


\pgraph{Price of Anarchy}
We quantify the welfare loss from valuation compressions by means of the {\em Price of Anarchy (PoA)}. 

The PoA with respect to PNE for valuations $v \in V$ is defined as the worst ratio between the optimal social welfare $\OPT(v)$ and the welfare $\SW(b)$ of a PNE $b \in B$,
\[\PoA(v) = \max_{b: \ \text{PNE}} \frac{\OPT(v)}{\SW(b)}\enspace .\]

Similarly, the PoA with respect to MNE, CE, and CCE for valuations $v \in V$ is the worst ratio between the optimal social welfare $\SW(b)$ and the expected welfare $\E_{b\sim \mathcal{B}}[\SW(b)]$ of a MNE, CE, or CCE $\mathcal{B}$, \[\PoA(v) = \max_{\mathcal{B}: \ \text{MNE}, \ \text{CE} \ \text{or} \ \text{CCE}} \frac{\OPT(v)}{\E_{b\sim \mathcal{B}}[\SW(b)]}\enspace .\]

We require that the bids $b_i$ for a given valuation $v_i$ are {\em conservative}, i.e., $b_i(S) \le v_i(S)$ for all bundles $S \subseteq M$. Similar assumptions are made and economically justified in the related work \cite{ChristodoulouKovacsSchapira08,BhawalkarRoughgarden11,FeldmanFuGravinLucier13}.

\section{Hardness Result for PNE with Additive Bids}\label{sec:hardness}
Our first result is that deciding whether there exists a pure Nash equilibrium of the VCG mechanism for restrictions from subadditive valuations to additive bids is $\mathcal{NP}$-hard. The proof of this result, which is given in \appref{app:hardness}, is by reduction from \textsc{3-Partition} \cite{GareyJohnson79} and uses an example with no pure Nash equilibrium from \cite{BhawalkarRoughgarden11}. The same decision problem is simple for $V\subseteq \XOS$ because pure Nash equilibria are guaranteed to exist \cite{ChristodoulouKovacsSchapira08}.

\begin{theorem}\label{thm:np}
Suppose that $V = \CF$, $B = \OS$, that the VCG mechanism is used, and that agents bid conservatively. Then it is $\mathcal{NP}$-hard to decide whether there exists a PNE.
\end{theorem}

\section{Smoothness Notion and Extension Results}\label{sec:smoothness}

Next we define a smoothness notion for mechanisms. It is weaker in some aspects and stronger in another aspect than the weak smoothness notion in \cite{SyrgkanisTardos13}. It is weaker because it allows agent $i$'s deviating bid $a_i$ to depend on the marginal distribution $\mathcal{B}_{-i}$ of the bids $b_{-i}$ of the agents other than $i$. This gives us more power in choosing the deviating bid, which might lead to better bounds. It is stronger because it does \emph{not} allow agent $i$'s deviating bid $a_i$ to depend on his own bid $b_i$. This allows us to prove bounds that extend to coarse correlated equilibria and not just correlated equilibria.  

\begin{definition}\label{def:smoothness}
A mechanism is relaxed $(\lambda,\mu_1,\mu_2)$-smooth for $\lambda,\mu_1,\mu_2 \ge 0$ if for every valuation profile $v \in V$, every distribution over bids $\mathcal{B}$, and every agent $i$ there exists a bid $a_i(v,\mathcal{B}_{-i})$ such that
\begin{multline*}
	\sum_{i \in N} \E_{b_{-i}\sim\mathcal{B}_{-i}}[u_i((a_i,b_{-i}),v_i)] \ge \lambda \cdot \OPT(v) - \mu_1 \cdot \sum_{i \in N} \E_{b\sim\mathcal{B}}[p_i(X_i(b),b_{-i})] - \mu_2 \cdot \sum_{i \in N} \E_{b\sim\mathcal{B}}[b_i(X_i(b))].
\end{multline*}
\end{definition}

\begin{theorem}\label{thm:cce}
If a mechanism is relaxed $(\lambda,\mu_1,\mu_2)$-smooth, then the Price of Anarchy under conservative bidding with respect to coarse correlated equilibria is at most
\[
\frac{\max\{\mu_1,1\}+\mu_2}{\lambda}.
\]
\end{theorem}
\begin{proof}
Fix valuations $v$. Consider a coarse correlated equilibrium $\mathcal{B}$. For each $b$ from the support of $\mathcal{B}$ denote the allocation for $b$ by $X(b) = (X_1(b),\dots,X_n(b))$. Let $a = (a_1,\dots,a_n)$ be defined as in \defref{def:smoothness}. Then,
\begin{align*}
	\E_{b\sim\mathcal{B}}[\text{SW}(b)]
	&= \sum_{i \in N} \E_{b\sim\mathcal{B}}[u_i(b,v_i)] + \sum_{i \in N} \E_{b\sim\mathcal{B}}[p_i(X_i(b),b_{-i})]\displaybreak[0]\\
	&\ge \sum_{i \in N} \E_{b_{-i}\sim\mathcal{B}_{-i}}[u_i((a_i,b_{-i}),v_i)] + \sum_{i \in N} \E_{b\sim\mathcal{B}}[p_i(X_i(b),b_{-i})]\displaybreak[0]\\
	&\ge \lambda\  \OPT(v) - (\mu_1 - 1) \sum_{i \in N} \E_{b\sim\mathcal{B}}[p_i(X_i(b),b_{-i})] - \mu_2  \sum_{i \in N} \E_{b\sim\mathcal{B}}[b_i(X_i(b))], 
\end{align*}
where the first equality uses the definition of $u_i(b,v_i)$ as the difference between $v_i(X_i(b))$ and $p_i(X_i(b),b_{-i})$, the first inequality uses the fact that $\mathcal{B}$ is a coarse correlated equilibrium, and the second inequality holds because $a=(a_1,\dots,a_n)$ is defined as in \defref{def:smoothness}.

Since the bids are conservative this can be rearranged to give
\begin{align*}
	(1+\mu_2) \E_{b\sim\mathcal{B}}[\text{SW}(b)]
	&\ge \lambda\  \OPT(v) - (\mu_1 - 1) \sum_{i \in N} \E_{b\sim\mathcal{B}}[p_i(X_i(b),b_{-i})]. 
\end{align*}

For $\mu_1 \le 1$ the second term on the right hand side is lower bounded by zero and the result follows by rearranging terms. For $\mu_1 > 1$ we use that $\E_{b\sim\mathcal{B}}[p_i(X_i(b),b_{-i})] \le \E_{b\sim\mathcal{B}}[v_i(X_i(b))]$ to lower bound the second term on the right hand side and the result follows by rearranging terms.
\end{proof}

\begin{theorem}\label{thm:learning}
If a mechanism is relaxed $(\lambda,\mu_1,\mu_2)$-smooth and $(b^1,\dots,b^T)$ is a sequence of conservative bids with vanishing external regret, then
\[
\frac{1}{T} \sum_{t=1}^{T} \SW(b^t) \ge \frac{\lambda}{\max\{\mu_1,1\}+\mu_2}\cdot \OPT(v) - o(1).
\]
\end{theorem}
\begin{proof}
Fix valuations $v$. Consider a sequence of bids $b^1,\dots,b^T$ with vanishing average external regret. For each $b^t$ in the sequence of bids denote the corresponding allocation by $X(b^t) = (X_1(b^t),\dots,X_n(b^t))$. Let $\delta^t_i(a_i) = u_i(a_i,b^t_{-i},v_i) - u_i(b^t,v_i)$ and let $\Delta(a) = \frac{1}{T} \sum_{t=1}^{T} \sum_{i=1}^{n} \delta^t_i(a_i)$. Let $a = (a_1,\dots,a_n)$ be defined as in \defref{def:smoothness}, where $\mathcal{B}$ is the empirical distribution of bids. Then,
\begin{align*}
	\frac{1}{T} \sum_{t=1}^{T} \text{SW}(b^t)
	&= \frac{1}{T} \sum_{t=1}^{T} \sum_{i=1}^{n} u_i(b^t_i,b^t_{-i},v_i) + \frac{1}{T} \sum_{t=1}^{T} \sum_{i=1}^{n} p_i(X_i(b^t),b^{t}_{-i})\\ \displaybreak[0]
	&= \frac{1}{T} \sum_{t=1}^{T} \sum_{i=1}^{n} u_i(a_i,b^t_{-i},v_i) + \frac{1}{T} \sum_{t=1}^{T} \sum_{i=1}^{n} p_i(X_i(b^t),b^{t}_{-i}) - \Delta(a)\\ \displaybreak[0]
	&\ge \lambda\ \OPT(v) - (\mu_1 - 1) \frac{1}{T} \sum_{t=1}^{T} \sum_{i=1}^{n} p_i(X_i(b^t,b^{t}_{-i})) - \mu_2  \frac{1}{T} \sum_{t=1}^{T} \sum_{i=1}^{n} b_i(X_i(b^t)) - \Delta(a),
\end{align*}
where the first equality uses the definition of $u_i(b^t_i,b^t_{-i},v_i)$ as the difference between $v_i(X_i(b^t))$ and $p_i(X_i(b^t),b^{t}_{-i})$, the second equality uses the definition of $\Delta(a)$, and the third inequality holds because $a=(a_1,\dots,a_n)$ is defined as in \defref{def:smoothness}.

Since the bids are conservative this can be rearranged to give
\begin{align*}
	(1+\mu_2)\ \frac{1}{T} \sum_{t=1}^{T} \text{SW}(b^t)
	&\ge \lambda\ \OPT(v) - (\mu_1 - 1) \frac{1}{T} \sum_{t=1}^{T} \sum_{i=1}^{n} p_i(X_i(b^t),b^{t}_{-i}) - \Delta(a).
\end{align*}

For $\mu_1 \le 1$ the second term on the right hand side is lower bounded by zero and the result follows by rearranging terms provided that $\Delta(a) = o(1)$. For $\mu_1 > 1$ we use that $\frac{1}{T} \sum_{t=1}^{T} \sum_{i=1}^{n} p_i(X_i(b^t),b^{t}_{-i}) \le \frac{1}{T} \sum_{t=1}^{T} \sum_{i=1}^{n} v_i(X_i(b^t))$ to lower bound the second term on the right hand side and the result follows by rearranging terms provided that $\Delta(a) = o(1)$.

The term $\Delta(a)$ is bounded by $o(1)$ because the sequence of bids $b^1,\dots,b^T$ incurs vanishing average external regret and, thus,
\begin{align*}
	\Delta(a)
	\le \frac{1}{T} \sum_{i=1}^{n} \left[ \max_{b'_i} \sum_{t=1}^{T} u_i(b'_i,b^t_{-i},v_i) - \sum_{t=1}^{T} u_i(b^t,v_i) \right]
	\le \frac{1}{T} \sum_{i=1}^{n} o(T). \tag*{\qedhere}
\end{align*}
\end{proof}

\section{Upper Bounds for CCE and Minimization of External Regret for Additive Bids}\label{sec:constant-bound}
We conclude our analysis of restrictions to additive bids by showing how the argument of \cite{FeldmanFuGravinLucier13} can be adopted to show that for restrictions from $V= CF$ to $B = OS$ the VCG mechanism is relaxed $(1/2,0,1)$-smooth. Using \thmref{thm:cce} we obtain an upper bound of $4$ on the Price of Anarchy with respect to coarse correlated equilibria. Using \thmref{thm:learning} we conclude that the average social welfare for sequences of bids with vanishing external regret converges to at least $1/4$ of the optimal social welfare. We thus improve the best known bounds by a logarithmic factor.

\begin{proposition}\label{prop:additive}
Suppose that $V = \CF$ and that $B = \OS$. Then the VCG mechanism is relaxed $(1/2,0,1)$-smooth under conservative bidding. 
\end{proposition}

To prove this result we need two auxiliary lemmata.

\begin{lemma}\label{lem:one}
Suppose that $V = \CF$, that $B = \OS$, and that the VCG mechanism is used. Then for every agent $i$, every bundle $Q_i$, and every distribution $\mathcal{B}_{-i}$ on the bids $b_{-i}$ of the agents other than $i$ there exists a conservative bid $a_i$ such that
\[
	\E_{b_{-i} \sim \mathcal{B}_{-i}}[u_i((a_i,b_{-i}),v_i)]  \ge \frac{1}{2} \cdot v_i(Q_i) - \E_{b_{-i} \sim \mathcal{B}_{-i}}[p_i(Q_i,b_{-i})]\enspace .
\]
\end{lemma}
\begin{proof}
Consider bids $b_{-i}$ of the agents $-i$. The bids $b_{-i}$ induce a price $p_i(j) = \max_{k \neq i} b_k(j)$ for each item $j$. Let $T$ be a maximal subset of items from $Q_i$ such that $v_i(T) \le p_i(T)$. Define the {\em truncated} prices $q_i$ as follows:
\begin{align*}
	q_i(j) = \begin{cases}
	               p_i(j) & \text{for $j \in Q_i \setminus T$, and}\\
	               0 & \text{otherwise}.
	               \end{cases}
\end{align*}

The distribution $\mathcal{B}_{-i}$ on the bids $b_{-i}$ induces a distribution $\mathcal{C}_i$ on the prices $p_i$ as well as a distribution $\mathcal{D}_i$ on the truncated prices $q_i.$

We would like to allow agent $i$ to draw his bid $b_i$ from the distribution $\mathcal{D}_i$ on the truncated prices $q_i$. For this we need that (1) the truncated prices are additive and that (2) the truncated prices are conservative. The first condition is satisfied because additive bids lead to additive prices. To see that the second condition is satisfied assume by contradiction that for some set $S \subseteq Q_i \setminus T$, $q_i(S) > v_i(S)$. As $p_i(S) = q_i(S)$ it follows that
\[
	v_i(S \cup T) \le v_i(S) + v_i(T) \le p_i(S) + p_i(T) = p_i(S \cup T),
\]
which contradicts our definition of $T$ as a maximal subset of $Q_i$ for which $v_i(T) \le p_i(T).$

Consider an arbitrary bid $b_i$ from the support of $\mathcal{D}_i$. Let $X_i(b_i,p_i)$ be the set of items won with bid $b_i$ against prices $p_i$. Let $Y_i(b_i,q_i)$ be the subset of items from $Q_i$ won with bid $b_i$ against the truncated prices $q_i$. As $p_i(j) = q_i(j)$ for $j \in Q_i \setminus T$ and $p_i(j) \ge q_i(j)$ for $j \in T$ we have $Y_i(b_i,q_i) \subseteq X_i(b_i,p_i) \cup T$. Thus, using the fact that $v_i$ is subadditive, $v_i(Y_i(b_i,q_i)) \le v_i(X_i(b_i,p_i)) + v_i(T).$ By the definition of the prices $p_i$ and the truncated prices $q_i$ we have $p_i(Q_i) - q_i(Q_i) = p_i(T) \ge v_i(T).$ By combining these inequalities we obtain 
\begin{align*}
v_i(X_i(b_i,p_i)) + p_i(Q_i) \ge v_i(Y_i(b_i,q_i)) + q_i(Q_i).
\end{align*}
Taking expectations over the prices $p_i \sim \mathcal{C}_i$ and the truncated prices $q_i \sim \mathcal{D}_i$ gives
\begin{align*}
	\E_{p_i \sim \mathcal{C}_i}[v_i(X_i(b_i,p_i)) + p_i(Q_i)]
	\ge \E_{q_i \sim \mathcal{D}_i}[v_i(Y_i(b_i,q_i)) + q_i(Q_i)].
\end{align*}

Next we take expectations over $b_i \sim \mathcal{D}_i$ on both sides of the inequality. Then we bring the $p_i(Q_i)$ term to the right and the $q_i(Q_i)$ term to the left. Finally, we exploit that the expectation over $q_i \sim \mathcal{D}_{i}$ of $q_i(Q_i)$ is the same as the expectation over $b_i \sim \mathcal{D}_{i}$ of $b_i(Q_i)$ to obtain
\begin{align}
	&\E_{b_i \sim \mathcal{D}_i}[\E_{p_i \sim \mathcal{C}_i}[v_i(X_i(b_i,p_i))]] - \E_{b_i \sim \mathcal{D}_i}[b_i(Q_i)] 
	\ge \E_{b_i \sim \mathcal{D}_i}[\E_{q_i \sim \mathcal{D}_i}[v_i(Y_i(b_i,q_i))]] - \E_{p_i \sim \mathcal{C}_i}[p_i(Q_i)]
	\label{eq:monkey}
\end{align}

Now, using the fact that $b_i$ and $q_i$ are drawn from the same distribution $\mathcal{D}_i$, we can lower bound the first term on the right-hand side of the preceding inequality by
\begin{align}
	\E_{b_i \sim \mathcal{D}_i}[\E_{q_i \sim \mathcal{D}_i}[v_i(Y_i(b_i,q_i)]]
	= \frac{1}{2} \cdot \E_{b_i \sim \mathcal{D}_i}[\E_{q_i \sim \mathcal{D}_i}[v_i(Y_i(b_i,q_i)) + v_i(Y_i(q_i,b_i))]] 
        \ge \frac{1}{2} \cdot v_i(Q_i),
	\label{eq:banana}
\end{align}
where the inequality in the last step comes from the fact that the subset $Y_i(b_i,q_i)$ of $Q_i$ won with bid $b_i$ against prices $q_i$ and the subset $Y_i(q_i,b_i)$ of $Q_i$ won with bid $q_i$ against prices $b_i$ form a partition of $Q_i$ and, thus, because $v_i$ is subadditive, it must be that $v_i(Y_i(b_i,q_i))+v_i(Y_i(q_i,b_i)) \ge v_i(Q_i).$

Note that agent $i$'s utility for bid $b_i$ against bids $b_{-i}$ is given by his valuation for the set of items $X_i(b_i,p_i)$ minus the price $p_i(X_i(b_i,p_i))$. Note further that the price $p_i(X_i(b_i,p_i))$ that he faces is at most his bid $b_i(X_i(b_i,p_i))$. Finally note that his bid $b_i(X_i(b_i,p_i))$ is at most $b_i(Q_i)$ because $b_i$ is drawn from $\mathcal{D}_i$. Together with inequality (\ref{eq:monkey}) and inequality (\ref{eq:banana}) this shows that
\begin{align*}
	\E_{b_i \sim \mathcal{D}_i} [\E_{b_{-i} \sim \mathcal{B}_{-i}}[u_i((b_i,b_{-i}),v_i)]] 
	\ge \E_{b_i \sim \mathcal{D}_i} [\E_{p_i \sim \mathcal{C}_{i}}[v_i(X_i(b_i,p_i)) - b_i(Q_i)]]
	\ge \frac{1}{2} \cdot  v_i(Q_i) - \E_{p_i \sim \mathcal{C}_i}[p_i(Q_i)].
\end{align*}

Since this inequality is satisfied in expectation if bid $b_i$ is drawn from distribution $\mathcal{D}_i$ there must be a bid $a_i$ from the support of $\mathcal{D}_i$ that satisfies it.
\end{proof}

\begin{lemma}\label{lem:two}
Suppose that $V = \CF$, that $B = \OS$, and that the VCG mechanism is used. Then for every partition $Q_1,\dots,Q_n$ of the items and all bids $b$, 
\[
	\sum_{i\in N} p_i(Q_i,b_{-i}) \le \sum_{i \in N} b_i(X_i(b)).
\]
\end{lemma}
\begin{proof}
For every agent $i$ and each item $j \in Q_i$ we have $p_i(j,b_{-i}) = \max_{k \neq i} b_k(j)\linebreak \le \max_{k} b_k(j)$. Hence an upper bound on the sum $\sum_{i\in N} p_i(Q_i,b_{-i})$ is given by $\sum_{i \in N} \max_k b_k(j)$. The VCG mechanisms selects allocation $X_1(b),\dots,X_n(b)$ such that $\sum_{i \in N} b_i(X_i(b))$ is maximized. The claim follows.
\end{proof}

\begin{proof}[Proof of \propref{prop:additive}]
The claim follows by applying \lemref{lem:one} to every agent~$i$ and the corresponding optimal bundle $O_i$, summing over all agents~$i$, and using \lemref{lem:two} to bound $\E_{b_{-i}\sim\mathcal{B}_{-i}}[\sum_{i \in N} p_i(O_i,b_{-i})]$ by $\E_{b\sim\mathcal{B}}[\sum_{i \in N} b_i(X_i(b))]$.
\end{proof}

An important observation is that the proof of the previous proposition requires that the class of price functions, which is induced by the class of bidding functions via the formula for the VCG payments, is contained in $B$. While this is the case for additive bids that lead to additive (or ``per item'') prices this is {\em not} the case for more expressive bids. In fact, as we will see in the next section, even if the bids are from OXS, the least general class from the hierarchy of \cite{LehmannLehmannNisan05} that strictly contains the class of additive bids, then the class of price functions that is induced by $B$ is no longer contained in $B$. This shows that the techniques that led to the results in this section {\em cannot} be applied to the more expressive bids that we study next.


\section{A Lower Bound for PNE with Non-Additive Bids}\label{sec:separation}

We start our analysis of non-additive bids with the following separation result: While for restrictions from subadditive valuations to additive bids the bound is $2$ for pure Nash equilibria \cite{BhawalkarRoughgarden11}, we show that for restrictions from subadditive valuations to OXS bids the corresponding bound is at least $2.4$. This shows that more expressiveness can lead to strictly worse bounds.

\begin{theorem}\label{thm:separation}
Suppose that $V = \CF$, that $\OXS \subseteq B \subseteq \XOS$, and that the VCG mechanism is used. Then for every $\delta > 0$ there exist valuations $v$ such that the PoA with respect to PNE under conservative bidding is at least $2.4 - \delta$.
\end{theorem}

The proof of this theorem makes use of the following auxiliary lemma, whose proof is deferred to \appref{app:subadditive}.

\begin{lemma}\label{lem:subadditive}
If $b_i \in \XOS$, then, for any $X \subseteq M$, 
\[
	\max_{S \subseteq X, |S| = |X|-1} b_i(S) \ge \frac{|X|-1}{|X|} \cdot b_i(X)\enspace .
\]
\end{lemma}

\begin{proof}[Proof of \thmref{thm:separation}]
There are $2$ agents and $6$ items. The items are divided into two sets $X_1$ and $X_2$, each with $3$ items. The valuations of agent $i \in \{1,2\}$ are given by (all indices are modulo two)
\begin{align*}
v_i(S) = \begin{cases}
	 12&\text{for $S \subseteq X_i$, $|S| = 3$}\\
	 6 &\text{for $S \subseteq X_i$, $1\leq|S| \le 2$}\\
	 5 + 1 \epsilon &\text{for $S \subseteq X_{i+1}$, $|S| = 3$}\\
	 4 + 2 \epsilon &\text{for $S \subseteq X_{i+1}$, $|S| = 2$}\\
	 3 + 3 \epsilon &\text{for $S \subseteq X_{i+1}$, $|S| = 1$}\\
	 \max_{j \in \{1,2\}}\{v_i(S\cap X_j)\}&\text{otherwise.}
	 \end{cases}
\end{align*}
The variable $\epsilon$ is a sufficiently small positive number. The valuation $v_i$ of agent $i$ is subadditive, but not fractionally subadditive. (The problem for agent $i$ is that the valuation for $X_i$ is too high given the valuations for $S \subset X_i$.)

The welfare maximizing allocation awards set $X_1$ to agent $1$ and set $X_2$ to agent $2$. The resulting welfare is $v_1(X_1)+v_2(X_2) = 12+12 = 24.$ 

We claim that the following bids $b = (b_1,b_2)$ are contained in $\OXS$ and constitutes a pure Nash equilibrium:
\begin{align*}
b_i(S) = \begin{cases}
	 0&\text{for $S \subseteq X_i$}\\
	 5 + 1 \epsilon &\text{for $S \subseteq X_{i+1}$, $|S| = 3$}\\
	 4 + 2 \epsilon &\text{for $S \subseteq X_{i+1}$, $|S| = 2$}\\
	 3 + 3 \epsilon &\text{for $S \subseteq X_{i+1}$, $|S| = 1$}\\
	 \max_{j \in \{1,2\}}\{b_i(S\cap X_j)\}&\text{otherwise.}
	 \end{cases}
\end{align*}

Given $b$ VCG awards set $X_2$ to agent $1$ and set $X_2$ to agent $2$ for a welfare of $v_1(X_2)+v_2(X_1) =  2 \cdot (5 + \epsilon)= 10 + 2 \epsilon$, which is by a factor $2.4-12\epsilon/(25+5\epsilon)$ smaller than the optimum welfare.

We can express $b_i$ as ORs of XORs of XS bids as follows: Let $X_i = \{a,b,c\}$ and $X_{i+1} = \{d,e,f\}$. Let $h_d, h_e, h_f$ and $\ell_d, \ell_e, \ell_f$ be $XS$ bids that value $d$, $e$, $f$ at $3+3\epsilon$ and $1-\epsilon$, respectively. Then $b_i(T) = (h_d(T) \otimes h_e(T) \otimes h_f(T)) \vee \ell_d(T) \vee \ell_e(T) \vee \ell_f(T)$.

To show that $b$ is a Nash equilibrium we can focus on agent $i$ (by symmetry) and on deviating bids $a_i$ that win agent $i$ a subset $S$ of $X_i$ (because agent $i$ currently wins $X_{i+1}$ and $v_i(S) = \max\{v_i(S\cap X_1),v_i(S\cap X_2)\}$ for sets $S$ that intersect both $X_1$ and $X_2$).

Note that the price that agent $i$ faces on the subsets $S$ of $X_i$ are superadditive: For $|S| = 1$ the price is $(5+\epsilon) - (4+2\epsilon) = 1-\epsilon$, for $|S| = 2$ the price is $(5+\epsilon) - (3+3\epsilon) = 2-2\epsilon$, and for $|S| = 3$ the price is $5+\epsilon.$

{\em Case 1:} $S = X_i$. 
We claim that this case cannot occur. To see this observe that because $a_i \in \XOS$, Lemma~\ref{lem:subadditive} shows that there must be a 2-element subset $T$ of $S$ for which $a_i(T) \ge 2/3 \cdot a_i(S)$. On the one hand this shows that $a_i(S) \le 9$ because otherwise $a_i(T) \ge 2/3 \cdot a_i(S) > 6$ in contradiction to our assumption that $a_i$ is conservative. On the other hand to ensure that VCG assigns $S$ to agent $i$ we must have $a_i(S) \ge a_i(T) + (3+3\epsilon)$
due to the subadditivity of the prices. Thus $a_i(S) \ge 2/3 \cdot a_i(S) + (3 + 3 \epsilon)$ and, hence, $a_i(S) \ge 9(1+\epsilon)$. We conclude that $9 \ge a_i(S) \ge 9(1+\epsilon)$, which gives a contradiction.

{\em Case 2:} $S \subset X_i$.
In this case agent $i$'s valuation for $S$ is $6$ and his payment is at least $1 - \epsilon$ as we have shown above. Thus, $u_i(a_i,b_{-i}) \leq 5 + \epsilon = u_i(b_i,b_{-i})$, i.e., the utility does not increase with the deviation.
\end{proof}

\section{Upper Bounds for CCE and Minimization of External Regret for Non-Additive Bids}\label{sec:non-additive}

Our next group of results concerns upper bounds for the PoA for restrictions to non-additive bids. For $\beta$-fractionally subadditive valuations we show that the VCG mechanism is relaxed $(1/\beta,1,1)$-smooth. By \thmref{thm:cce} this implies that the Price of Anarchy with respect to coarse correlated equilibria is at most $2\beta$. By \thmref{thm:learning} this implies that the average social welfare obtained in sequences of repeated play with vanishing external regret converges to $1/(2\beta)$ of the optimal social welfare. For subadditive valuations, which are $O(\log(m))$-fractionally subadditive, we thus obtain bounds of $O(\log(m))$ resp.~$\Omega(1/\log(m))$. For fractionally subadditive valuations, which are $1$-fractionally subadditive, we thus obtain bounds of $2$ resp.~$1/2$. We thus extend the results of \cite{ChristodoulouKovacsSchapira08,BhawalkarRoughgarden11} from additive to non-additive bids.

\begin{proposition}\label{prop:nonadditive}
Suppose that $V \subseteq \beta\text{-}XOS$ and that $\OS \subseteq B \subseteq \XOS$, then the VCG mechanism is relaxed $(1/\beta,1,1)$-smooth under conservative bidding.
\end{proposition}

We will prove that the VCG mechanism satisfies the definition of relaxed smoothness point-wise. For this we need two auxiliary lemmata.

\begin{lemma}\label{lem:tight}
Suppose that $V \subseteq \beta\text{-}\XOS$, that $\OS \subseteq B \subseteq \XOS$, and that the VCG mechanism is used. Then for all valuations $v \in V$, every agent $i$, and every bundle of items $Q_i \subseteq M$ there exists a conservative bid $a_i \in B_i$ such that for all conservative bids $b_{-i} \in B_{-i}$,
\[
	u_i(a_i,b_{-i},v_i) \ge \frac{v_i(Q_i)}{\beta} - p_i(Q_i,b_{-i}).
\]
\end{lemma}
\begin{proof}
Fix valuations $v$, agent $i$, and bundle $Q_i$. As $v_i \in \beta\text{-}\XOS$ there exists a conservative, additive bid $a_i \in \OS$ such that $\sum_{j \in X_i} a_i(j) \le v_i(X_i)$ for all $X_i \subseteq Q_i$, and $\sum_{j \in Q_i} a_i(j) \ge \frac{v_i(Q_i)}{\beta}$.
Consider conservative bids $b_{-i}$. Suppose that for bids $(a_i,b_{-i})$ agent $i$ wins items $X_i$ and agents $-i$ win items $M\setminus X_i$. As VCG selects outcome that maximizes the sum of the bids, 
\begin{align*}
	a_i(X_i) + b_{-i}(M\setminus X_i) &\ge a_i(Q_i) + b_{-i}(M\setminus Q_i).
\end{align*}
We have chosen $a_i$ such that $a_i(X_i) \le v_i(X_i)$ and $a_i(Q_i) \ge v_i(Q_i)/\beta$. Thus,
\begin{align*}
	v_i(X_i) + b_{-i}(M\setminus X_i) 
	&\ge a_i(X_i) + b_{-i}(M\setminus X_i)
	\ge a_i(Q_i) + b_{-i}(M\setminus Q_i) \ge \frac{v_i(Q_i)}{\beta} + b_{-i}(M\setminus Q_i).
\end{align*}
Subtracting $b_{-i}(M)$ from both sides gives
\begin{equation*}
	v_i(X_i) - p_i(X_i,b_{-i}) 
	\ge \frac{v_i(Q_i)}{\beta} - p_i(Q_i,b_{-i}).
\end{equation*}
As $u_i((a_i,b_{-i}),v_i) = v_i(X_i) - p_i(X_i,b_{-i})$ this shows that $u_i((a_i,b_{-i}),v_i) \ge v_i(Q_i)/\beta - p_i(Q_i,b_{-i})$ as claimed.
\end{proof}

\begin{lemma}\label{lem:close}\label{lem:prices}
Suppose that $\OS \subseteq B \subseteq \XOS$ and that the VCG mechanism is used. For every allocation $Q_1,\dots,Q_n$ and all conservative bids $b \in B$ and corresponding allocation $X_1,\dots,X_n$,
\[
	\sum_{i=1}^{n} [p_i(Q_i,b_{-i}) - p_i(X_i,b_{-i})] \le \sum_{i=1}^{n} b_i(X_i) \enspace.
\]
\end{lemma}
\begin{proof}
We have $p_i(Q_i,b_{-i}) = b_{-i}(M) - b_{-i}(M\setminus Q_i)$ and $p_i(X_i,b_{-i}) = b_{-i}(M) - b_{-i}(M\setminus X_i)$ because the VCG mechanism is used. Thus,
\begin{align}
	\sum_{i=1}^{n} [p_i(Q_i,b_{-i}) - p_i(X_i,b_{-i})] 
        &= \sum_{i=1}^{n} [b_{-i}(M\setminus X_i) - b_{-i}(M\setminus Q_i)]. \label{eq:sun}
\end{align}
We have $b_{-i}(M\setminus X_i) = \sum_{k \neq i} b_k(X_k)$ and $b_{-i}(M\setminus Q_i) \ge \sum_{k \neq i} b_k(X_k \cap (M\setminus Q_i))$ because $(X_k \cap (M \setminus Q_i))_{i \neq k}$ is a feasible allocation of the items $M \setminus Q_i$ among the agents $-i$. Thus,
\begin{align}
        \sum_{i=1}^{n} [b_{-i}(M\setminus X_i) - b_{-i}(M\setminus Q_i)]
	&\le \sum_{i=1}^{n} [\sum_{k \neq i} b_k(X_k) - \sum_{k \neq i} b_k(X_k \cap (M\setminus Q_i))] \notag\\
	&\le \sum_{i=1}^{n} [\sum_{k=1}^{n} b_k(X_k) - \sum_{k=1}^{n} b_k(X_k \cap (M\setminus Q_i))] \notag \\
        &= \sum_{i=1}^{n} \sum_{k=1}^{n} b_k(X_k) - \sum_{i=1}^{n} \sum_{k=1}^{n} b_k(X_k \cap (M\setminus Q_i)). \label{eq:moon}
\end{align}
The second inequality holds due to the monotonicity of the bids.
Since $\XOS = 1\text{-}\XOS$ for every agent $k$, bid $b_k \in \XOS$, and set $X_k$ there exists a bid $a_{k,X_k} \in \OS$ such that $b_k(X_k) = a_{k,X_k}(X_k) = \sum_{j \in X_k} a_{k,X_k}(j)$ and $b_k(X_k \cap (M \setminus Q_i)) \ge a_{k,X_k}(X_k \cap (M \setminus Q_i)) = \sum_{j \in X_k \cap (M \setminus Q_i)} a_{k,X_k}(j)$ for all $i$.
As $Q_1, \dots, Q_n$ is a partition of $M$ every item is contained in exactly one of the sets $Q_1, \dots, Q_n$ and hence in $n-1$ of the sets $M\setminus Q_1, \dots, M\setminus Q_n.$ By the same argument for every agent $k$ and set $X_k$ every item $j \in X_k$ is contained in exactly $n-1$ of the sets $X_k \cap (M \setminus Q_1), \dots, X_k \cap (M \setminus Q_n)$. Thus, for every fixed $k$ we have that $ \sum_{i=1}^{n} b_k(X_k \cap (M\setminus Q_i)) \ge (n-1) \cdot \sum_{j \in X_k} a_{k,X_k}(j) = (n-1) \cdot a_{k,X_k}(X_k) = (n-1) \cdot b_k(X_k)$. It follows that
\begin{align}
	\sum_{i=1}^{n} \sum_{k=1}^{n} b_k(X_k) - &\sum_{i=1}^{n} \sum_{k=1}^{n} b_k(X_k \cap (M\setminus Q_i)) \notag\displaybreak[0]\\
        &\le n \cdot \sum_{k=1}^{n} b_k(X_k) - (n-1) \cdot \sum_{k=1}^{n} b_k(X_k) 
        = \sum_{i=1}^{n} b_k(X_k). \label{eq:stars}
\end{align}
The claim follows by combining inequalities (\ref{eq:sun}), (\ref{eq:moon}), and (\ref{eq:stars}).
\end{proof}

\begin{proof}[Proof of \propref{prop:nonadditive}]
Applying \lemref{lem:tight} to the optimal bundles $O_1,\dots,O_n$ and summing over all agents $i$,
\[
	\sum_{i \in N} u_i(a_i,b_{-i},v) \ge \frac{1}{\beta} \OPT(v) - \sum_{i \in N} p_i(O_i,b_{-i}) \enspace.
\]
Applying \lemref{lem:prices} we obtain
\[
	\sum_{i \in N} u_i(a_i,b_{-i},v) \ge \frac{1}{\beta} \OPT(v) - \sum_{i \in N} p_i(X_i(b),b_{-i}) - \sum_{i \in N} b_i(X_i(b)). \tag*{\qedhere}
\]
\end{proof}

\section{More Lower Bounds for PNE with Non-Additive Bids}\label{sec:lower-bounds}

We conclude by proving matching lower bounds for the VCG mechanism and restrictions from fractionally subadditive valuations to non-additive bids.  We prove this result by showing in \appref{app:outcome-closure} that the VCG mechanism satisfies the {\em outcome closure property} of \cite{Milgrom10}, which implies that when going from more general bids to less general bids no new pure Nash equilibria are introduced. Hence the lower bound of $2$ for pure Nash equilibria and additive bids of \cite{ChristodoulouKovacsSchapira08} translates into a lower bound of $2$ for pure Nash equilibria and non-additive bids. 
\begin{theorem}
Suppose that $\OXS \subseteq V \subseteq \CF$, that $\OS \subseteq B \subseteq \XOS$, and that the VCG mechanism is used. Then the PoA with respect to PNE under conservative bidding is at least $2$.
\end{theorem}
Note that the previous result applies even if valuation and bidding space coincide, and the VCG mechanism has an efficient, dominant-strategy equilibrium. This is because the VCG mechanism admits other, non-efficient equilibria and the Price of Anarchy metric does not restrict to dominant-strategy equilibria if they exist.

\bibliographystyle{abbrvnat}
\bibliography{abbshort,literature}

\appendix


\section{Results for Incomplete Information Setting}\label{app:incomplete}

Denote the distribution from which the valuations are drawn by $\mathcal{D}$. Possibly randomized bidding strategies $b_i: V_i \rightarrow B_i$ form a {\em mixed Bayes-Nash equilibrium (MBNE)} if for every agent $i \in N$, every valuation $v_i \in V_i$, and every pure deviation $b'_i \in B_i$
 \[
	\E_{v_{-i} \sim \mathcal{D}_{-i}}[u_i(b_i,b_{-i},v_i)] \ge \E_{v_{-i} \sim \mathcal{D}_{-i}}[u_i(b'_i,b_{-i},v_i)].
\]

The PoA with respect to mixed Bayes-Nash equilibria for a distribution over valuations is the ratio between the expected optimal social welfare and the expected welfare of the worst mixed Bayes-Nash equilibrium
\[
	\PoA = \max_{b: \ \mathrm{MBNE}} \frac{\E_{v\sim\mathcal{D}}[\OPT(v)]}{\E_{v\sim\mathcal{D}}[\SW(b)]}
\]

\begin{theorem}\label{thm:bayes}
If an auction is relaxed $(\lambda,\mu_1,\mu_2)${-smooth} then the Price of Anarchy with respect to MBNE under conservative bidding is at most
\[
	\frac{\max\{\mu_1,1\} + \mu_2}{\lambda}.
\]
\end{theorem}
\begin{proof}
Fix a distribution $\mathcal{D}$ on valuations $v$. Consider a mixed Bayes-Nash equilibrium $\mathcal{B}$ and denote the allocation for bids $b$ by $X(b) = (X_1(b),\dots,X_n(b))$. Let $a = (a_1,\dots,a_n)$ be defined as in \defref{def:smoothness}. Then, 
\newcommand{\Exp}{\mathop{\E_{v \sim D}}}
\begin{align*}
	\Exp[\text{SW}(b)]
	&= \sum_{i=1}^{n} \hspace{-4pt} \E_{v_i \sim D_i} [\mathop{\E_{v_{-i} \sim D_{-i}}}[u_i((b_i,b_{-i}),v_i)]] + \hspace{-8pt} \Exp[\sum_{i=1}^{n} p_i(X_i(b),b_{-i})]]\\ \displaybreak[0]
	&\ge \sum_{i=1}^{n} \hspace{-4pt} \E_{v_i \sim D_i} [\mathop{\E_{v_{-i} \sim D_{-i}}}[u_i((a_i,b_{-i}),v_i)]] + \hspace{-8pt} \Exp[\sum_{i=1}^{n} p_i(X_i(b),b_{-i})]\\ \displaybreak[0]
	&\ge \lambda \cdot OPT(v) - (\mu_1-1) \Exp[\sum_{i=1}^{n} p_i(X_i(b),b_{-i})] - \mu_2 \cdot \Exp[\sum_{i=1}^{n} b_i(X_i(b))],
\end{align*}
where the first equality uses the definition of $u_i((b_i,b_{-i}),v_i)$ as the difference between $v_i(X_i(b))$ and $p_i(X_i(b),b_{-i})$, the first inequality uses the fact that $\mathcal{B}$ is a mixed Bayes-Nash equilibrium, and the second inequality uses that $a=(a_1,\dots,a_n)$ is defined as in \defref{def:smoothness}.

Since the bids are conservative this can be rearranged to give
\begin{align*}
	(1+\mu_2) \Exp[\text{SW}(b)]
	&\ge \lambda \cdot OPT(v) - (\mu_1-1) \Exp[\sum_{i=1}^{n} p_i(X_i(b),b_{-i})].
\end{align*}

For $\mu_1 \le 1$ the second term on the right hand side is lower bounded by zero and the result follows by rearranging terms. For $\mu_1 > 1$ we use that $\E_{v\sim\mathcal{D}}[p_i(X_i(b),b_{-i})] \le \E_{v\sim\mathcal{D}}[v_i(X_i(b))]$ to lower bound the second term on the right hand side and the result follows by rearranging terms.
\end{proof}


\section{Proof of Theorem \ref{thm:np}}\label{app:hardness}
Given an instance of \textsc{3-Partition} consisting of a multiset of $3m$ positive integers $w_1, \dots, w_{3m}$ that sum up to $mB$, we construct an instance of a combinatorial auction in which the agents have subadditive valuations in polynomial time as follows:

The set of agents is $B_1,\dots,B_m$ and $C_1,\dots,C_m$. The set of items is $\mathcal{I}\cup \mathcal{J}$, where $\mathcal{I}=\{I_1,\dots,I_{3m}\}$ and $\mathcal{J}=\{J_1,\dots,J_{3m}\}$. Let $\mathcal{J}_i=\{J_i,J_{m+i},J_{2m+i}\}$.
Every agent $B_i$ has valuations
\begin{align*}
	v_{B_i}(S)&=\max\{v_{\mathcal{I},B_i}(S), v_{\mathcal{J},B_i}(S)\}, \qquad &&\text{where}\\
	v_{\mathcal{I},B_i}(S)&=\sum_{e\in \mathcal{I}\cap S} w_e, &&\text{and}\\
	v_{\mathcal{J},B_i}(S)&=
	\begin{cases}
	10B&\text{if}\ |\mathcal{J}_i\cap S|=3,\\
	5B&\text{if}\ |\mathcal{J}_i\cap S|\in \{1,2\},\\
	0&\text{otherwise.}
	\end{cases}
\intertext{Every agent $C_i$ has valuations}
	v_{C_i}(S)&=
	\begin{cases}
	16B&\text{if}\ |\mathcal{J}_i\cap S|=3,\\
	8B&\text{if}\ |\mathcal{J}_i\cap S|\in \{1,2\},\\
	0&\text{otherwise.}
\end{cases}
\end{align*}
The valuations for the items in $\mathcal{J}$ are motivated by an example for valuations without a PNE in \cite{BhawalkarRoughgarden11}.
Note that our valuations are subadditive.

We show first that if there is a solution of our \textsc{3-Partition} instance then the corresponding auction has a PNE.
Let us assume that $P_1,\dots,P_m$ is a solution of \textsc{3-Partition}. We obtain a PNE when every agent $B_i$ bids $w_j$ for each $I_j$ with $j\in P_i$ and zero for the other items; and every agent $C_i$ bids $4B$ for each item in $\mathcal{J}_i$.
The first step is to show that no agent $B_i$ would change his strategy. The utility of $B_i$ is $B$, because $B_i$'s payment is zero.
As the valuation function of $B_i$ is the maximum of his valuation for the items in $\mathcal{I}$ and the items in $\mathcal{J}$ we can study the strategies for $\mathcal{I}$ and $\mathcal{J}$ separately.
If $B_i$ would change his bid and win another item in $\mathcal{I}$, $B_i$ would have to pay his valuation for this item because there is an agent $B_j$ bidding on it, and, thus, his utility would not increase. As $B_i$ bids conservatively, $B_i$ could win at most one item of the items in $\mathcal{J}_i$. His value for the item would be $5B$, but the payment would be $C_i$'s bid of $4B$. Thus, his utility would not be larger than $B$ if $B_i$ would win an item of $\mathcal{J}$. Hence, $B_i$ does not want to change his bid.
The second step is to show that no agent $C_i$ would change his strategy. This follows since the utility of every agent $C_i$ is $16B$, and this is the highest utility that $C_i$ can obtain.

We will now show two facts that follow if the auction is in a PNE:
(1) We first show that in every PNE every agent $B_i$ must have a utility of at least $B$. To see this denote the bids of agent $C_i$ for the items in $\mathcal{J}_i$ by $c_1$, $c_2$, and $c_3$ and assume w.l.o.g.~that $c_1\leq c_2\leq c_3$. As agent $C_i$ bids conservatively, $c_2+c_3\leq 8B$, and, thus, $c_1\leq 4B$. If agent $B_i$ would bid $5B$ for $c_1$, $B_i$ would win $c_1$ and his utility would be at least $B$, because $B_i$ has to pay $C_i$'s bid for $c_1$. As $B_i$'s utility in the PNE cannot be worse, his utility in the PNE has to be at least $B$. (2) Next we show that in a PNE agent $B_i$ cannot win any of the items in $\mathcal{J}_i$. For a contradiction suppose that agent $B_i$ wins at least one of the items in $\mathcal{J}_i$ by bidding $b_1$, $b_2$, and $b_3$ for the items in $\mathcal{J}_i$. Then agent $C_i$ does not win the whole set $\mathcal{J}_i$ and his utility is at most $8B$. As agent $B_i$ bids conservatively, $b_i+b_j \le 5B$ for $i \neq j \in \{1,2,3\}$. Then, $b_1+b_2+b_3 \le 7.5B$. Agent $C_i$ can however bid $b_1+\epsilon$, $b_2+\epsilon$, $b_3+\epsilon$ for some $\epsilon > 0$ without violating conservativeness to win all items in $\mathcal{J}_i$ for a utility of at least $16B - 7.5B > 8B$. Thus, $C_i$'s utility increases when $C_i$ changes his bid, i.e., the auction is not in a PNE.

Now we use fact (1) and (2) to show that our instance of \textsc{3-Partition} has a solution if the auction has a PNE. Let us assume that the auction is in a PNE. By (1) we know that every agent $B_i$ gets at least utility $B$. Furthermore, by (2) we know that every agent $B_i$ wins only items in $\mathcal{I}$. It follows that every agent $B_i$ pays zero and has exactly utility $B$. Thus, the assignment of the items in $\mathcal{I}$ corresponds to a solution of \textsc{3-Partition}.


\section{Proof of Lemma \ref{lem:subadditive}}\label{app:subadditive}
As $b_i \in \XOS$ there exists an additive bid $a_i$ such that $\sum_{j \in X} a_i(j) = b_i(X)$ and for every $S \subseteq X$ we have $b_i(S) \ge \sum_{j \in S} a_i(j).$ There are $|X|$ many ways to choose $S \subseteq X$ such that $|S| = |X|-1$ and these $|X|$ many sets will contain each of the items $j \in X$ exactly $|X| - 1$ times. Thus, 
$	\sum_{S \subseteq X, |S| = |X|-1} b_i(S) \ge (|X|-1) \cdot b_i(X).$
For any set $T \in \arg \max_{S \subseteq X, |S| = |X|-1} b_i(S)$, using the fact that the maximum is at least as large as the average, we therefore have $b_i(T) \ge (|X|-1)/|X| \cdot b_i(X)$. 


\section{Outcome Closure}\label{app:outcome-closure}
We say that a mechanism satisfies {\em outcome closure} for a given class $V$ of valuation functions and a restriction of the class $B$ of  bidding functions to a subclass $B'$ of bidding functions if for every $v \in V$, every $i$, every conservative $b'_{-i} \in B'$, and every conservative $b_i \in B$ there exists a conservative $b'_i \in B'$ such that $u_i(b'_i,b'_{-i},v_i) \ge u_i(b_i,b'_{-i},v_i)$.

\begin{proposition}\label{prop:outcome-closure}
If a mechanism satisfies outcome closure for a given class $V$ of valuation functions and a restriction of the class $B$ of bidding functions to a subclass $B'$, then the Price of Anarchy with respect to pure Nash equilibria under conservative bidding for $B$ is at least as large as for $B'$.
\end{proposition}
\begin{proof}
It suffices to show that the set of PNE for $B'$ is contained in the set of PNE for $B$. To see this assume by contradiction that, for some $v \in V$, $b' \in B'$ is a PNE for $B'$ but not for $B$. As $b'$ is not a PNE for $B$ there exists an agent $i$ and a bid $b_i \in B$ such that $u_i(b_i,b'_{-i},v_i) > u_i(b'_i,b'_{-i},v_i)$. By outcome closure, however, there must be a bid $b''_i \in B'$ such that $u_i(b''_i,b'_{-i},v_i) \ge u_i(b_i,b'_{-i},v_i)$. It follows that $u_i(b''_i,b'_{-i},v_i) > u_i(b'_i,b'_{-i},v_i)$, which contradicts our assumption that $b'$ is a PNE for $B'$.
\end{proof}

Next we use outcome closure to show that the Price of Anarchy in the VCG mechanism with respect to pure Nash equilibria weakly increases with expressiveness for classes of bidding functions below XOS.

\begin{proposition}\label{the:weakly-increasing}
Suppose that $V \subseteq \CF$, that $B' \subseteq B \subseteq \XOS$, and that the VCG mechanism is used. Then the Price of Anarchy with respect to pure Nash equilibria under conservative bidding for $B$ is at least as large as for $B'$.
\end{proposition}
\begin{proof}
By Proposition~\ref{prop:outcome-closure} it suffices to show that the VCG mechanism satisfies outcome closure for $V$ and the restriction of $B$ to $B'$. 
For this fix valuations $v \in V$, bids $b'_{-i} \in B'$, and consider an arbitrary bid $b_i \in B$ by agent $i$. Denote the bundle that agent $i$ gets under $(b_i,b'_{-i})$ by $X_i$ and denote his payment by $p_i = p_i(X_i,b'_{-i}).$ Since $b_i \in B \subseteq \XOS$ there exists a bid $b'_i \in \OS \subseteq B'$ such that
\begin{align*}
        &\sum_{j \in X_i} b'_i(j) = b_i(X_i) &&\text{and,}\\
	&\sum_{j \in S} b'_i(j) \le b_i(S) &&\text{for all $S \subseteq X_i$.}
\end{align*}

By setting $b'_i(j) = 0$ for $j \not\in X_i$ we ensure that $b'_i$ is conservative. Recall that the VCG mechanism assigns agent $i$ the bundle of items that maximizes his reported utility. We have that $b'_i(X_i) = b_i(X_i)$ and that $b'_i(T) \le b_i(T)$ for all $T \subseteq M$. We also know that the prices $p_i(T,b_{-i})$ for all $T \subseteq M$ do not depend on agent $i$'s bid. Hence agent $i$'s reported utility for $X_i$ under $b'$ is as high as under $b$ and his reported utility for every other bundle $T$ under $b'$ is no higher than under $b$. This shows that agent $i$ wins bundle $X_i$ and pays $p_i$ under bids $(b'_i,b'_{-i})$.
\end{proof}


\end{document}